\documentclass[aps,pra,showpacs,preprint,groupedaddress]{revtex4}
\usepackage{graphicx}
\usepackage{amssymb}
\usepackage{amsmath}
\usepackage{amsthm}
\usepackage{latexsym}
\usepackage{dcolumn}
\usepackage{bm}

\newtheorem{thm}{Theorem}[section]

\newtheorem{lem}{Lemma}[section]

\newcommand{\0}{| 0 \rangle}
\newcommand{\1}{| 1 \rangle}
\newcommand{\g}{| G \rangle}
\newcommand{\gr}{| G_R \rangle}
\newcommand{\gl}{| G_L \rangle}
\newcommand{\qs}[1]{| #1 \rangle}

\begin{document}

\title{Efficient One-way Quantum Computations for Quantum Error Correction}

\author{W. Huang$^{1,*}$,Zhaohui Wei$^2$}

\affiliation{$^1$Department of Electrical Engineering and Computer
Science, University of Michigan, Ann Arbor, MI 48109 \\
$^2$State Key Laboratory of
Intelligent Technology and Systems, Department of Computer Science
and Technology, Tsinghua University, Beijing, China, 100084}

\email{weihuang@eecs.umich.edu}

\begin{abstract}
We show how to explicitly construct an $O(nd)$ size and constant
quantum depth circuit which encodes any given $n$-qubit stabilizer
code with $d$ generators. Our construction is derived using the
graphic description for stabilizer codes and the one-way quantum
computation model. Our result demonstrates how to use cluster states
as scalable resources for many multi-qubit entangled states and how
to use the one-way quantum computation model to improve design of
quantum algorithms.

\end{abstract}


 \maketitle

\section{Introduction}

The \textit{one-way quantum computation (1WQC)} model
\cite{RB01,RB02,RBB03}, due to its simplicity, universality and
parallelism, is widely considered as a very promising scheme for the
experimental development of a quantum computer
\cite{N04,ND04,BR04,CMJ04,BK04,Nature05,BES05}.

The 1WQC model starts with a highly entangled cluster state and
performs quantum computing simply by a sequence of adaptive
single-qubit measurements and post-measurement local corrections.
Thus the whole computation is separated into four parts: 1)
preparing cluster states, 2) performing single-qubit measurements,
3) classically processing measurement outcomes, and 4) performing
post-measurement local unitary corrections.

Such a simple model has been proved to be universal for quantum
computation since any quantum circuit can be efficiently simulated
on it. Moreover, by translating normal quantum circuits into
1WQC-compatible circuits, it is possible to reduce circuit depth and
increase parallelism, which is critical to overcome the quantum
decoherence problem \cite{J05,BK07}.

In this paper, we reproduce the previous encoding procedure for
quantum error correction \cite{CG96} under one-way quantum
computation model. Using only $O(nd)$ single-bit operations and a
small amount of two-bit measurements, we encode any given $n$-qubit
stabilizer code with $d$ generators. Furthermore, we will show that
the depth of our construction is constant. Our construction is
derived using the graphic description for stabilizer codes.

This paper is organized as follows: in Section II, we review the
connections between graph states and stabilizer codes. In Section
III, we produce an $O(nd)$ size and constant-depth 1WQC-compatible
circuit for the encoding and decoding procedure of arbitrary graph
codes.

\section {Preliminaries}

In this section, let us recall some basic notions concerning this
paper. More details can be reviewed in \cite{G97, NC00}.

\subsection{Stabilizer codes, graph codes and graph states}

The \textit{Pauli group} $\mathcal{P}_n$ on $n$ qubits is defined to
consist of $n$-fold tensor products of the Pauli matrices
\{I,X,Y,Z\}
\[
  I=\left(
    \begin{array}{cc}
      1 & 0 \\
      0 & 1 \\
    \end{array}
  \right) ,
  X=\left(
    \begin{array}{cc}
      0 & 1 \\
      1 & 0 \\
    \end{array}
  \right) ,
  Y=\left(
    \begin{array}{cc}
      0 & -i \\
      i & 0 \\
    \end{array}
  \right) ,
  Z=\left(
    \begin{array}{cc}
      1 & 0 \\
      0 & 1 \\
    \end{array}
  \right)
\] with multiplicative factors \{$\pm 1$, $\pm \textit{i}$\}.

The \textit{Clifford group}   $\mathcal{C}_n$ is defined as the
normalizer of the Pauli group $\mathcal{P}_n$,

\begin{equation*}
    \mathcal{C}_n = \{U \in SU(2^n)~~| ~~\forall P\in \mathcal{P}_n ,~~~  UPU^\dagger \in
    \mathcal{P}_n
    \}.
\end{equation*}

The \textit{local Clifford group} $\mathcal{LC}_n = \mathcal{C}_1
^{\otimes n}$ is a subgroup of $\mathcal{C}_n$ which only consists
of the tensor products of local unitary operations. $The Clifford
group \mathcal{C}_n$ can be generated, up to a global phase factor,
by the Hadamard gate \textit{H}, the phase gate \textit{S} and the
CNOT gate, while the local Clifford group $\mathcal{LC}_n$ can be
generated by the \textit{H} and \textit{S} only. There are only 24
elements in $\mathcal{LC}_1$, up to a global phase.

\textit{Stabilizer group} $\mathcal{S}_n$ is an Abelian subgroup of
\textit{Pauli group} $\mathcal{P}_n$ without $-I$. Any
$\mathcal{S}_n$ is a stabilizer for a non-trivial vector space,
which can be defined as the codespace of a \textit{Stabilizer code},
The codewords of the Stabilizer code form the +1-eigenspace of the
all the operations in $\mathcal{S}_n$.

A $n$-qubit \textit{Stabilizer code} with $d$ generators can encode
a $(n-d)$-qubit state into a $n$-qubit state. The stabilizer group
$\mathcal{S}_n$ corresponding to the stabilizer code can be
generated by $d$ independent elements $g_1,...,g_d$. Other elements
in $\mathcal{S}_n$ can be represented as products of $g_1,...,g_d$.
Thus we can use $g_1,...,g_d$ to describe a stabilizer code. We will
use the \textit{binary framework} of stabilizer formalism to
represent elements in stabilizer group efficiently.

Define a homomorphic map from $\mathcal{P}_1$ to $\mathbb{Z}^2_2$ as
the following:
\[
 I \rightarrow 00,~~ X \rightarrow 10,~~ Y \rightarrow 11,~~ Z \rightarrow
 01 .
\]
After mapping, an element of Pauli group $P = P_1\otimes P_1\otimes
... P_N $ can be described by a binary vector $(X|Z)$,where X is the
vector consisting of the first bits of $P_1,...,P_n$ while Z is the
vector consisting of the second bits. Therefore a $n$-qubit
stabilizer code $C$ with $d$ generators $g_1,...,g_d$ can be
described by a $2n \times d $ \textit{check matrix}: $ S = [X|Z]$
where both X and Z are $n \times d $ matrices.

A $n$-qubit \textit{stabilizer state} $|\psi\rangle$ is a $n$-qubit
stabilizer code with exactly $n$ generators. In this case, the
dimension of the code space is one. $|\psi\rangle$ is the only
vector stabilized by $n$ generators, up to a global phase. The
stabilizer of $|\psi\rangle$ can be described by a $2n \times n $
check matrix. A \textit{graph state} $|G\rangle$ is a stabilizer
state with graphical check matrix $[X|Z] = [I|G]$, where $G$ is the
adjacency matrix of the underlying graph of the graph state. a
$n$-qubit \textit{graph code} with $d$ generators is generated by
check matrix $S[X|Z] = [
                                                               \begin{array}{cc}
                                                                 I ,& R \\
                                                               \end{array}
                                                               | \begin{array}{cc} A+RC^T ,& C \\
                                                               \end{array}
]$, which is closely related to the graph state stabilized by
$S[X|Z] = [
                                                               \begin{array}{cc}
                                                                 I & 0 \\
                                                                 0& I \\
                                                               \end{array}
                                                               | \begin{array}{cc} A & C \\
                                                                                   C^T & 0 \\
                                                               \end{array}
]$.

An important result \cite {S02,GKR01,NDM04} is that any stabilizer
state can be transformed into a graph state with generator matrix
$S[X|Z] = [I|G]$ by a local unitary operation U $\in
\mathcal{LC}_n$. Similarly, any $n$ qubits stabilizer code with $d$
generators can be transformed into a graph code with generator
matrix $S[X|Z] = [
                                                               \begin{array}{cc}
                                                                 I ,& R \\
                                                               \end{array}
                                                               | \begin{array}{cc} A+RC^T ,& C \\
                                                               \end{array}
]$ related to the graph state
 stabilized by $S[X|Z] = [
                                                               \begin{array}{cc}
                                                                 I & 0 \\
                                                                 0& I \\
                                                               \end{array}
                                                               | \begin{array}{cc} A & C \\
                                                                                   C^T & 0 \\
                                                               \end{array}
]$. The adjacency matrix of the underlying graph is
$
\left (
\begin{array}{cc} A & C \\
                                                                                   C^T & 0 \\
                                                               \end{array}
                                                               \right
                                                               )
$.

\subsection{Relation between graph states
and graph codes}

In this subsection, we give an example to demonstrate how to
generate graph codes based on graph state. According to the
relationship between stabilizer codes and graph codes, the basic
idea can be generalized to find the relationship between stabilizer
codes and graph states.

Suppose we have a six-qubit graph code with four generators
$\{g'_1,g'_2,g'_3,g'_4\}$ which is stabilized by
\[
 \left(
  \begin{array}{cccc}
    1 & 0 & 0 & 0 \\
    0 & 1 & 0 & 0 \\
    0 & 0 & 1 & 0 \\
    0 & 0 & 0 & 1 \\
  \end{array}
    \begin{array}{cccc}
    0 & 1 \\
    0 & 0  \\
    1 & 0  \\
    1 & 0  \\
  \end{array}
\right | \left .
  \begin{array}{cccccc}
    0 & 1 & 1 & 1 & 1 & 0 \\
    1 & 0 & 1 & 0 & 0 & 0 \\
    1 & 1 & 0 & 0 & 0 & 1 \\
    0 & 0 & 1 & 1 & 1 & 0 \\
  \end{array}
\right )
 \]

Suppose we also have a graph state $|G \rangle$ with generators
$\{g_1,g_2,g_3,g_4,g_5,g_6\}$ which is stabilized by
\[
 \left(
  \begin{array}{cccccc}
    1 & 0 & 0 & 0 & 0 & 0 \\
    0 & 1 & 0 & 0 & 0 & 0 \\
    0 & 0 & 1 & 0 & 0 & 0 \\
    0 & 0 & 0 & 1 & 0 & 0 \\
    0 & 0 & 0 & 0 & 1 & 0 \\
    0 & 0 & 0 & 0 & 0 & 1 \\
  \end{array}
\right | \left .
  \begin{array}{cccccc}
    0 & 1 & 0 & 1 & 1 & 0 \\
    1 & 0 & 1 & 0 & 0 & 0 \\
    0 & 1 & 0 & 1 & 0 & 1 \\
    1 & 0 & 1 & 0 & 1 & 0 \\
    1 & 0 & 0 & 1 & 0 & 0 \\
    0 & 0 & 1 & 0 & 0 & 0 \\
  \end{array}
\right )
\]

It's not hard to see following equations relating the two sets of
generators:
\begin{equation*}
   g'_1 = g_1 \cdot g_6 ~~ , ~~
   g'_2 = g_2 ~~,  ~~
   g'_3 = g_3 \cdot g_5 ~~, ~~
   g'_4 = g_4 \cdot g_5 ~~.
\end{equation*}

Based on the above relation between the generators of the graph code
and the graph state, we can obtain the graphic representation of the
graph code as shown in the Fig.1.

Let $G'$ denote the graph which includes $G$ plus input nodes $A$
and $B$. Let $|G'\rangle$ denote the graph state corresponding to
$G'$. Suppose the codewords of the graph code are \{ $
|00_L\rangle$, $|01_L\rangle$, $|10_L\rangle$, $|11_L\rangle$ \},
then we have:
\begin{multline*}
  |G'\rangle  = |0\rangle_A|0\rangle_B|00_L\rangle +
   |0\rangle_A|1\rangle_B|01_L\rangle \\
   +|1\rangle_A|0\rangle_B|10_L\rangle  + |1\rangle_A|1\rangle_B|11_L\rangle
\end{multline*}
and
\begin{align*}
 & |00_L\rangle = |G\rangle  \\ &  |01_L\rangle =Z_3Z_4Z_5|G\rangle
\\&  |10_L\rangle =Z_1Z_6|G\rangle  \\& |11_L\rangle
=Z_1Z_3Z_4Z_5Z_6|G\rangle
\end{align*}
where $Z_i$ denotes local unitary $Z$ on qubit $i$.

Fault-tolerant $X$ and $Z$ operations on the first and second qubit
of the encoded state are $(X_L)_1 = Z_1Z_6$, $(X_L)_2 = Z_3Z_4Z_6$,
$(Z_L)_1 = g_1 = X_1Z_2Z_5$ and $(Z_L)_2 = g_3 = X_3Z_2Z_4Z_6$. More
details about graph states will be explained in the next section.

\textbf{Remark:} if we can construct the uniform encoded state
$\sum\limits_{x\in \{0,1\}^k}|x\rangle|x_L\rangle$ in a stabilizer
code, then we can encode any given unknown $k$-qubit state
$|\psi\rangle$ in a stabilizer code by quantum teleportation
\cite{NC00}. In the next section, we focus on using cluster states
to generate any graph state including uniform encoded states for
graph codes.

\begin{figure}
\includegraphics[width=5.0cm]{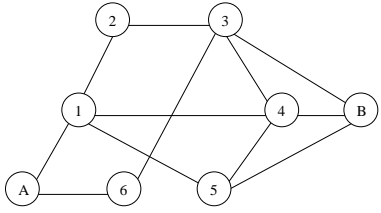}
\caption[Fig. 1.] { Graphic representation of the graph code
generated by matrix T. Vertices A and B are input nodes.}
\end{figure}

\section{Construction of 1WQC-compatible circuit encoding graph codes}

\subsection{Preparation of graph states by definition}

According to the section above, for the purpose of encoding, we have
to prepare the graph state we need. In this subsection, we focus on
this topic. Firstly, let us recall the definition of graph states.
Actually, this definition itself is a method of creating graph
states.

Let $G = (V,E)$ be a graph with $n=|V|$ vertices and $m=|E|$ edges,
then graph state $|G\rangle$ corresponding to the graph $G$ is the
following superposition over all basis states,
\begin{equation*}
|G\rangle = \prod\limits_{\substack{ i<j \\ (i,j)\in
E}}Z_{ij}|+\rangle^{\otimes n}= \sum\limits_{x \in \{0,1\}^n}
(-1)^{q(x)}|x\rangle.
\end{equation*}

Here $Z_{ij}$ denotes the controlled phase gate between qubit $i$
and qubit $j$.
\[
Z_{ij}|+\rangle_i|+\rangle_j = |0\rangle_i|+\rangle_j +
|1\rangle_i|-\rangle_j.
\]

q(x) is a quadratic function related to the graph $G$
\[
q(x)= \sum\limits_{\substack { i<j \\
(i,j)\in E }}x_ix_j
\].

We can verify that $|G\rangle$ is the stabilizer state with the
graphical check matrix $[X|Z] = [I|G]$. Thus we have the following
procedure of preparing the graph state $|G\rangle$ by its natural
definition:

\begin{itemize}
  \item the qubit at each vertex $v\in V$ has the initial state $|0\rangle$,
  \item apply the Hadamard gate on each qubit, so each qubit is now in the state $|+\rangle,$
  \item apply the controlled phase gate $Z_{ij}$ to each edge
  $(i,j)\in E$.
\end{itemize}

Actually, cluster states and graph states are used so widely in
quantum information processing that the preparation of them becomes
an important issue. Many efforts have been made on this problem. On
one hand, it has been shown that cluster states can be grown using a
'divide-and-conquer' approach \cite{Nielsen04,YR03,BR05,BK05,DD05}.
In this approach, bigger cluster states are created by iteratively
connecting smaller clusters together.

On the other hand, another scheme for the preparation of cluster
states is based on optical lattice of ultracold atoms
\cite{JJB99,MGWR03}. In this proposal, the cluster state can be
prepared in one step using a natural nearest-neighbour interaction.
Though this is a theoretical proposal at the present time because of
the difficulties in experiments, it may be a promising and efficient
method of preparing cluster states in the future. In this situation,
it seems necessary to propose a general method for preparing
arbitrary graph states from 2D clusters states. In the following, we
will give such a procedure. Firstly, let us recall some properties
about graph states in the next subsection.

\subsection{Graphical rules of single-qubit pauli measurements}

We start by describing some graphical rules of the operations on the
graph states .

Let $\lambda_a$ denote the \textit{local complement} operation on
vertex $a$ which replaces the subgraph induced by $a$'s neighbors
with its complement. Let $[X]_a$, $[Y]_a$, $[Z]_a$ denote single
qubit \textit{Pauli measurements} $X$, $Y$, $Z$ on qubit $a$
respectively. After each Pauli measurement, a graph state $|G
\rangle$ will transform into another graph state $|\widetilde{G}
\rangle$, up to a local Clifford unitary depending on the
measurement outcome \cite{HD06}. The graphical rules are the
following:
\begin{itemize}
  \item $[Z]_a$ : deletes vertex $a$ and related edges from $G$,
~~$\widetilde{G}=G-a$.
  \item $[Y]_a$ : first applies local complement on vertex $a$,
then delete vertex a,~~$\widetilde{G}=\lambda_a(G)-a$.
  \item $[X]_a$ : chooses any of $a$'s neighbor
$b$, applies local complement on vertex $b$, then applies local
complement on $a$, deletes vertex $a$ and applies local complement
on $b$ again, ~~$\widetilde{G}=\lambda_b(\lambda_a \circ
\lambda_b(G)-a)$.
\end{itemize}

In summary, two graph operations, local complement and vertex
deletion, can be achieved by single qubit Pauli measurements and
local Clifford operations. Based on above simple graphic rules, two
more graph operations, crossing and contraction, can be implemented
as the following:

\begin{figure}
\includegraphics[width=5.0cm]{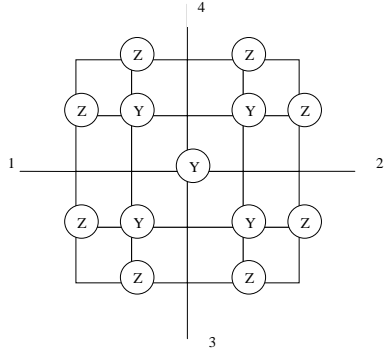}
\caption[Fig. 2.] { Two chains cross each other without sharing
qubits. }
\end{figure}

\begin{lem}(Crossing)
Two chains crossing each other without sharing any qubit can be
simulated in the cluster state of 2D lattice .
\end{lem}

\begin{proof}
As shown in Fig.2, chain crossing can be implemented by performing
some Pauli measurements in a $5\times5$ lattice cluster state. In
the first step, we perform all $Z$ measurements and correct the
related byproduct local unitaries. Then, we do all $Y$ measurements
except the central one and correct them. Finally, we perform the $Y$
measurement on the central qubit and correct it. The result is a
chain from 1 to 2 and another chain from 3 to 4. The two chains
cross each other without sharing any qubit.
\end{proof}

\textbf{Remark:} One can use the rewrite rules of the measurement
calculus \cite{DKP06}, to reduce the running time by postponing the
local corrections till the end of $1WQC$. However it is unnecessary
here since the running time of simulating chain crossing is already
constant.

\begin{figure}
\includegraphics[width=5.0cm]{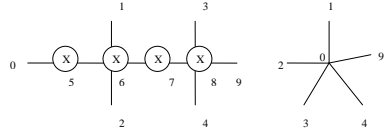}
\caption[Fig.3.] {Contraction through a chain}
\end{figure}

\begin{lem}(Contraction\cite{HMP06})\label{contraction}
Let graph G(L-v-a-b-R) consists of subgraph L and R and three
vertices $v$,$a$ and $b$. Vertex $a$ has two edges (a,v) and (a,b).
Vertices $v$ and $b$ have edges connected to the vertices in the
subgraphs $L$ and $R$ respectively. After applying $X$ measurements
on vertices $a$ and $b$, the graph state corresponding to graph $G$
(L-v-a-b-R) will change to the graph state corresponding to the
graph $\widetilde{G}$ (L-v-R).
\end{lem}

\begin{proof} The lemma can be verified by applying the basic graphic rule
about $[x]_a$. The contractions through a chain can be done
simultaneously. For example in Fig.3. if $X$ measurements are
applied on the qubits 5, 6, 7 and 8 at the same time, the graph will
contract to the vertex 0, whether or not a $Z$ operation on the
qubit 0 is needed for local correction depends on the sum of
measurement outcomes of the qubits 6 and 8. Local $Z$ corrections on
the qubits 1 and 2 depend on the measurement outcome of the qubit 5.
Local $Z$ corrections on the qubits 3, 4 and 9 depend on the sum of
measurement outcomes of qubit 5 and 7.

To understand why the contractions through a chain can be done
simultaneously and how local operations can be postponed to the end
of the computing, we have to go through some complicated
calculations step by step carefully. A detailed proof of the lemma
is included in the appendix.
\end{proof}

\subsection{Generating arbitrary graph states from the cluster states of 2D lattice}

\begin{figure}
\includegraphics[width=5.0cm]{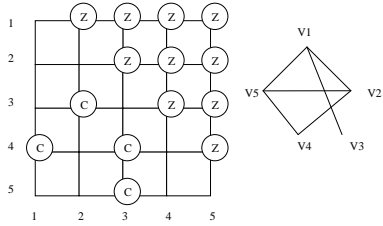}
\caption[Fig.4.] { Generating arbitrary graph states from the
cluster states of 2D lattice}
\end{figure}

\begin{thm}[Generating arbitrary graph state from the cluster states of 2D lattice]
Any graph state with the underlying graph $G$ can be generated from
a  $O(n) \times O(n) $ cluster state by local measurements and local
unitaries in constant time.
\end{thm}
\begin{proof}

Given a graph $G=(V,E)$ with $n$ vertices $v_1,...v_n$ and $m$ edges
$E_1,...,E_m$, we need to perform some crossings and contractions on
a cluster state of 2D lattice to generate a graph state $|G\rangle$.

We need some several auxiliary qubits. A $5 \times 5 $ lattice is
required for implementing crossings whereas contractions between any
two vertices with degree great than two requires degree two
auxiliary qubits. However, we can introduce those auxiliary qubits
by increasing the length and width of the 2D lattice only by a
constant value. Thus to simplify the proof, we ignore those
auxiliary qubits and only consider an $ n\times n$ lattice.

In the first step, we perform $Z$ measurements on the qubit located
at the intersection of the $i$th row and the $j$th column where
$i<j$ and $i,j\in \{1,2,...,n\}$. In the second step, we perform a
crossing operation for any qubit on the location $(i,j)$ which
satisfies $i>j$ and $(v_i,v_j)\notin E$. In the final step, we
contract simultaneously through the columns $1,2,...n$ to the
locations (1,1),(2,2),...,(n,n).

We show an example in Fig.4.
\end{proof}
In fact, since we only need to generate the specific graph states
related to the graph codes, by carefully rearranging the protocol in
Theorem III.1, it is not difficult to reduce the size of the cluster
state needed by our encoding method, to $O(n) \times O(d)$. (For
convenience, we introduce the above one)

\subsection{Construction of 1WQC encoding graph codes}

Combining separate pieces of the operations together, our encoding
algorithm is the following: given any $n$-qubit stabilizer code with
$d$ generators, we first determine the equivalent graph code and the
graphic representation of the graph code. Then we build the cluster
state of 2D lattice and generate the corresponding graph state from
the cluster state of 2D lattice. Finally, we encode any unknown
quantum state by quantum teleportation.

\section{Discussion}
According to the above section, the whole computation which
generates the graph state related to any $n$-qubit graph code with
$d$ generators can be conducted on an $O(n)\times O(d)$ lattice.
Therefore the total number of quantum operations of the 1WQC is
bounded by $O(nd)$, which is the length of the description of
generating matrix of the stabilizer code. Therefore, both the size
and the depth of our 1WQC are most likely optimal in general case.

Note that our construction has a constant running time. Since qubit
coherent time is limited, improving the temporal overhead of
encoding procedure will be helpful for its physical implementation.

Furthermore, it should be pointed out that in the procedure of
preparing graph state based on cluster states, most operations we
need are single-qubit operation (except when teleportating unknown
state, where a small mount of two-qubit measurements are involved).
Usually, in experiments the fidelity of one-qubit operations is very
high. Thus, ignoring errors introduced by one-qubit operations, our
encoding procedure will be reasonable as long as the quality of
cluster states we use as foundation is good enough.

The decoding procedure can be done in a similar way as the encoding
procedure, if one can implement the quantum teleportation on an
encoded state. For the error detecting and fault tolerant
computation on the states encoded in the stabilizer codes, one can
apply methods in \cite{G97}.

\section{Conclusion}

In this paper, we have shown how to use one-way quantum computation
to implement an encoding and decoding procedure for quantum error
correction. We have constructed an $O(nd)$ size and constant-depth
1WQC-compatible circuit which encodes any given $n$-qubit stabilizer
code with $d$ generators. The result demonstrates that the cluster
states can be used as the scalable resources for many multi-qubit
entangled states and the one-way quantum computation model can help
to design better quantum algorithms than the traditional quantum
circuit model.

\begin{acknowledgments}

W. Huang thanks Y. Y. Shi and Y.-J.~Han, Z. Wei thanks L. M. Duan
and Y. -J. Han for discussions about graph states. W. Huang would
also like to thank labmates D. R. Vandenberg and R. Duan for
interesting discussions about local complement operations. This work
was supported in part by The China Scholarship Council, the NSF
awards (0431476), the ARDA under ARO contracts and the A. P. Sloan
Fellowship.
\end{acknowledgments}

\onecolumngrid

\appendix*
\section{}

\begin{proof}[Proof of Lemma \ref{contraction}]

At the beginning, the initial state $\qs{\psi}$ is the graph state
$|G\rangle$.
\begin{align*}
\qs{\psi} & = \g = (\0 + Z_vZ_b \1  )_a (\0 + Z_R \1  )_b (\0 + Z_L
\1 )_v \gr \gl \\&= ((I+Z_vZ_b)\qs{+}+(I-Z_vZ_b)\qs{-})_a (\0 + Z_R
\1 )_b (\0 + Z_L \1 )_v \gr \gl
\end{align*}

Apply $X $ measurements on qubit a, let $x\in\{0,1\}$ be the
measurement result,the remaining state of other qubits is
\begin{align*}
\qs{\psi} & = (I+(-1)^xZ_vZ_b)(\0 + Z_R \1  )_b (\0 + Z_L \1 )_v
\gr\gl \\&= ((I+(-1)^xZ_v)\0 + (I-(-1)^xZ_v)Z_R\1)_b (\0 + Z_L \1
)_v \gr\gl \\&= ((I+(-1)^xZ_v+Z_R-(-1)^xZ_vZ_R)\qs{+} +
(I+(-1)^xZ_v-Z_R+(-1)^xZ_vZ_R)\qs{-})_b \\&(\0 + Z_L \1 )_v\gr\gl
\end{align*}

Then apply $X $ measurements on qubit b, let $y\in\{0,1\}$ be the
measurement result, let $Z^0 = I$ and $Z^1 = Z$,  we have
\begin{align*}
\qs{\psi} & = (I+(-1)^xZ_v+(-1)^yZ_R-(-1)^{x+y}Z_vZ_R)(\0 + Z_L \1
)_v\gr\gl \\&= ((1+(-1)^x)I+(-1)^y(1-(-1)^x)Z_R)\0
+((1-(-1)^x)I+(-1)^y(1+(-1)^x)Z_R)Z_L\1) \\& \gr\gl
\\&  =
\begin{cases}
(\0+(-1)^yZ_RZ_L\1)_v\gr\gl& \text{if x = 0 },\\
((-1)^yZ_R\0+Z_L\1)_v\gr\gl & \text{if x = 1}.
\end{cases}
\\&= Z_v^yZ_R^x(\0 + Z_RZ_L\1)_v\gr\gl
\\&= Z_v^yZ_R^x\qs{\widetilde{G}}
\end{align*}
\end{proof}
More generally, we consider the effect of graph contraction on
generalized graph state
$Z_L^{u_0}Z_v^{v_0}Z_R^{w_0}Z_a^{x_0}Z_b^{y_0}\g$, where
$\{u_0,v_0,w_0,x_0,y_0\} \in \{0,1\}$. In this case, we first apply
some local Z operations on graph state $\g$, then apply $X$
measurements on qubits a and b, denoted as $[X]_a$ and $[X]_b$ .
Suppose measurement results are x and y respectively. Applying $Z$
operation before X measurement $[X]$ on a qubit  does nothing but
flip the measurement outcome, therefore \begin{align*} \qs{\psi} & =
[X]_a[X]_bZ_L^{u_0}Z_v^{v_0}Z_R^{w_0}Z_a^{x_0}Z_b^{y_0}\g
\\&=
Z_L^{u_0}Z_v^{v_0}Z_R^{w_0}[X]_aZ_a^{x_0}[X]_bZ_b^{y_0}\g
\\&=Z_L^{u_0}Z_v^{v_0+y_0+y}Z_R^{w_0+x_0+x}\qs{\widetilde{G}}
\end{align*}

We can see that local Z operations on $a$ and $b$ pass to $R$ and
$v$ respectively after graph contraction.

\end{document}